%% file: persistent-rspa.tex
\documentclass[openacc]{rsproca_new}

\newtheorem{theorem}{\bf Theorem}[section]
\newtheorem{proposition}{\bf Proposition}[section]
\newtheorem{lemma}{\bf Lemma}[section]
\newtheorem{remark}{\bf Remark}[section]

\newtheorem{definition}{\bf Definition}[section]

\input{macros}

\usepackage{3dplot}
\usepackage{tikz-cd}
\usepackage{pgfplots}
\pgfplotsset{compat=1.14}
\tdplotsetmaincoords{60}{110}
\usetikzlibrary{shapes,arrows,positioning}
\usepgfplotslibrary{fillbetween}
\usetikzlibrary{intersections}
\pgfdeclarelayer{bg}
\pgfsetlayers{bg,main}

\begin{document}

\title{Persistent homology for low-complexity models}
\author{Martin Lotz}
\address{Mathematics Institute\\ The University of Warwick}

\subject{applied mathematics, computational geometry, topological data analysis}


\corres{Martin Lotz\\
\email{martin.lotz@warwick.ac.uk}}

\begin{abstract}
We show that recent results on randomized dimension reduction schemes that exploit structural properties of data can be applied in the context of persistent homology. In the spirit of compressed sensing, the dimension reduction is determined by the Gaussian width of a structure associated to the data set, rather than its size, and such a reduction can be computed efficiently. We further relate the Gaussian width to the doubling dimension of a finite metric space, which appears in the study of the complexity of other methods for approximating persistent homology. We can therefore literally replace the ambient dimension by an intrinsic notion of dimension related to the structure of the data.
\end{abstract}


\begin{fmtext}


\section{Introduction}
Persistent homology is an approach to topological data analysis (TDA) that allows to infer multi-scale qualitative information from noisy data. Starting from a {\em point cloud} representing the data, persistent homology extracts topological information about the structure from which the data is assumed to be sampled from (such as number of connected components, holes, cavities, ...) by associating multi-scale invariants, the {\em barcodes} or {\em persistence diagrams} to the data. These invariants measure topological features of neighbourhoods of the data at different scales; features that {\em persist} over large scale ranges are considered relevant, while short lived features are considered noise. 

Despite excellent theoretical guarantees and plenty of practical applications, a large number of data points ($n$) and dimension ($d$) can cause significant challenges to the computation of persistent homology. Much current work in the field is devoted to addressing this challenge, the underlying rationale being that the true complexity of the data is often smaller than it appears. 
Our focus is on the analysis of randomized dimension reduction schemes at the point cloud level that depend purely on structural properties of the data points, and not on the size of the data set. Specifically, we show that it is possible to approximate the persistent homology of a point cloud from its projection to a subspace of dimension proportional to a measure of intrinsic dimension, the {\em Gaussian width} of an underlying structure. 
\end{fmtext}

\maketitle

A consequence of our results is that if we assume that the $d$-dimensional data points are $s$-sparse (having at most $s$ non-negligible entries) in a suitable basis or frame (for example, image data in a Fourier or wavelet basis), then we can work in an ambient dimension of order $O(s\log(d/s))$; the dimension reduction is the same as that achievable for sparse signal recovery in compressed sensing. 

The Gaussian width is closely related to another intrinsic dimension parameter, the {\em doubling dimension} of a metric space. While previous work has shown that the doubling dimension can replace the ambient dimension in the analysis of various approaches to computing persistent homology, it follows that one can also literally embed the data into an ambient space of dimension proportional to it.
The dimension reduction affects the very first part of the persistent homology pipeline, where it can reduce the size of the input. Such a reduction is useful in constructions that depend on the ambient dimension, while the independence of the number of data points is useful in applications where the size of the data set is not known in advance or may change. In addition, we will see that the reduction can be computed efficiently under certain circumstances.

We point out that the notion of ``low complexity'' used here differs from the usual manifold assumption, where the data is assumed to lie close to a lower dimensional set whose topology one is interested in, and where it is the intrinsic dimension of that manifold that determines the complexity of the problem. In our setting, we do not make such an assumption, but only consider structural properties of any potential data points. These are related to the type of data we consider and can be known or estimated in advance.
To illustrate the difference between these notions of complexity, consider image data from the Columbia Object Image Library~\cite{nene1996columbia}, which contains photos of objects rotated around a fixed axis. The images associated to one single object lie on a one-dimensional structure (a circle). If the database would be extended to include more perspectives on each object, then that structure will change to a sphere. Independent of this, the individual images are images, and as such are compressible and lie close to a low-dimensional subspace arrangement. 
It is this latter structure that determines the target dimension of the random dimension reduction, regardless of the shape of the manifold around which the images cluster, and independent of the number of images present. 

As mentioned before, a crucial parameter in our context is the Gaussian width of a set $S$,
\begin{equation*}
  w(S) = \Expect[\sup_{\vct{x}\in S}\ip{\vct{x}}{\vct{g}}],
\end{equation*}
where the expectation is over a standard Gaussian vector $\vct{g}$.
The Gaussian width features prominently in the study of Gaussian processes~\cite{ledoux2013probability}, in geometric functional analysis~\cite{artstein2015asymptotic}, learning theory~\cite{Bartlett2003}, and in compressed sensing~\cite[Chapter 9]{Foucart2013}. We show that it also determines the dimension in which persistent topological information can be recovered from Euclidean point clouds without much loss. Formally, this means that the persistence diagrams for the original and for the projected data are close in some metric, which can be formalized using the interleaving distance on persistence modules. 
For the precise definition of these and other concepts used in the statement of the result, see Section~\ref{sec:persistent}.

\begin{theorem}\label{thm:mainA}
Let $X\subset M\subset \R^d$ with $X$ finite, let $\delta\in (0,1)$, and $T=\{(\vct{x}-\vct{y})/\norm{\vct{x}-\vct{y}}\mid \vct{x},\vct{y}\in M\}$. Assume that
\begin{equation*}
m \geq \frac{\left(w(T)+\sqrt{2\log(2/\delta)}\right)^2}{\e^2}+1.
\end{equation*}
Then for a random $m\times d$ matrix $\mtx{G}$ with normal distributed entries $g_{ij}\sim N(0,1/m)$,
with probability at least $1-\delta$, the persistence modules
associated to the \v{C}ech, Vietoris-Rips, and Delaunay complexes of $X$ and $\mtx{G}X$ with respect to the Euclidean distance are multiplicatively $(1-\e)^{-1}$-interleaved. 
\end{theorem}

Theorem~\ref{thm:mainA} is based on, and recovers as a special case, an extension of the Johnson-Lindenstrauss Theorem by Sheehy~\cite{sheehy2014persistent} (see the example with the Gaussian width of discrete sets below). One crucial difference to the classical approach is that the Gaussian width allows us to do better when the data $X$ has a particularly simple structure. 
Theorem~\ref{thm:mainA} should be seen as a prototype for a whole class of dimensionality reduction results and is, as stated, not practical. More practically relevant variants of Theorem~\ref{thm:mainA} are discussed in section Section~\eqref{sec:efficient}. 

Before proceeding, we present some examples of sets where the Gaussian width is well known. We use the notation $w^2(T) := w(T)^2$ for the square of the Gaussian width.\\

\textbf{Discrete set.} Let $T=\{\vct{x}_1,\dots,\vct{x}_n\}$ be a set of $n$ points with $\norm{\vct{x}_i}_2=1$ for $1\leq i\leq n$. Then
\begin{equation}\label{eq:w-discrete}
  w^2(T) \leq 2\log(n).
\end{equation}
A proof can be found in~\cite[Sections 2.5]{boucheron2013concentration}.\\

\textbf{Spheres and balls.} Let $T=S^{m-1}$ be an $(m-1)$-dimensional unit sphere in $\R^d$. Then the invariance property of the Gaussian distribution implies
\begin{equation*}
  w^2(T) = \Expect[\norm{\overline{\vct{g}}}]^2\leq m,
\end{equation*} 
where $\overline{\vct{g}}=(g_1,\dots,g_m)^T$ is the projection of a Gaussian vector $\vct{g}\in \R^d$ to the first $m$ coordinates.\\

\textbf{Sparse vectors.} Let $T_s = \{ \vct{x}\in S^{d-1} \mid |\supp{\vct{x}}|\leq s\}$ be the set of $s$-sparse unit vectors.
As shown by Rudelson and Vershynin~\cite{Rudelson2008}, the squared Gaussian width of this set is bounded by
\begin{equation}\label{eq:w-sparse}
  w^2(T_s) \leq C\cdot s\log(d/s),
\end{equation}
where $C$ is some constant. As the Gaussian width is orthogonally invariant, it is enough to require that the elements of $\vct{x}$ are sparse in some fixed basis. For example, we could have a collection of compressed images that are sparse in a discrete cosine or wavelet basis, or signals that are sparse in a frequency domain.\\ 

\textbf{Low-rank matrices.} Let $M_r = \{ \vct{X}\in \R^{d_1\times d_2} \mid \norm{\mtx{X}}_F=1, \ \rk(\mtx{X})\leq r\}$ be the set of matrices of rank at most $r$ and unit Frobenius norm, where $\norm{\mtx{X}}_F^2 = \sum_{i,j} X_{ij}^2$.
It can be shown that
\begin{equation*}
  w^2(M_r) \leq C \cdot r(d_1+d_2)
\end{equation*}
for some constant $C$, see~\cite{recht2010guaranteed,candes2011tight} for a derivation and more background. Examples of low-rank matrices or approximately low-rank matrices abound, including images, Euclidean distance matrices, correlation matrices, matrices arising from the discretization of differential equations, or recommender systems. One can also consider low-rank tensors (with respect to several notions of rank).\\

\textbf{Linear images.}
Assume that $T=\mtx{A}S$, where $S\subset \R^{d}$ and $\mtx{A}\in \R^{m\times d}$. Then the squared Gaussian width of $T$ can be bounded in terms of that of $S$ and the condition number $\kappa(\mtx{A})$ of $\mtx{A}$,
\begin{equation*}
  w^2(T)\leq \kappa^2(\mtx{A})w^2(S).
\end{equation*}
See~\cite{effective} for a derivation of this bound in a more general context. This is useful when considering the cosparse signal recovery setting~\cite{Nam2013}, in which the signals of interest are sparse after applying some (not necessarily invertible) linear transformation.\\

\textbf{Convex cones.} 
Let $T=C\cap S^{d-1}$, where $C$ is a convex cone (a convex set with $\lambda \vct{x}\in C$ if $\vct{x}\in C$ and $\lambda\geq 0$). The Gaussian width of $C\cap S^{d-1}$ differs from an invariant of the cone, the statistical dimension $\delta(C)$, by at most one~\cite[Prop 10.2]{edge}. It is known that $\delta(C) = d/2$ for self-dual cones (this includes the orthant and cone of positive semidefinite matrices), $\delta(C) \approx \log(d)$ for $C=\{\vct{x} \mid x_1\leq \cdots \leq x_d\}$, and $\delta(C) \approx d\sin^2(\alpha)$
for the circular cone of radius $\alpha$~\cite[Chapter 3]{edge}.
Moreover, approximations are known for the squared Gaussian width of the descent cones of the $1$-norm~\cite{Stojnic2009} and the nuclear norm~\cite{CRPW:12}, see also \cite[Chapter 4]{edge}.

\subsection{Considerations} We discuss some issues and extensions related to Theorem~\ref{thm:mainA}. These are concerned with efficiency, applications, limitations, and robustness.

\subsubsection{Efficiency}\label{sec:efficient}
In many applications, the computational cost of multiplying the data with a dense Gaussian matrix is likely to offset any potential gains of working in a lower dimension~\cite{sarlos2006improved}. In persistent homology, however, where a first step consists of computing the pairwise distances of $n$ points, the projection has to be computed only $n$ times, followed by $n(n-1)/2$ distance computations in a lower dimension. It follows that if the dimension is fixed and the number of samples is large enough, any projection
to a lower dimension will eventually lead to computational savings. That being said, recent results around the Johnson-Lindenstrauss Theorem can be used to extend Theorem~\ref{thm:mainA} (up to constants and logarithmic factors) to a large class of linear maps, including subgaussian matrices~\cite{dirksen2016dimensionality}, sparse Johnson-Lindenstrauss transforms~\cite{bourgain2015toward}, and matrices satisfying a classical Restricted Isometry Property~\cite{oymak2015isometric, krahmer2011new}. Strikingly, in~\cite{oymak2015isometric} the authors derived a ``transfer theorem" that shows that one can use so-called RIP (Restricted Isometry Property) matrices with only minor loss. Such matrices have been studied extensively in compressed sensing~\cite{Foucart2013}, and include the SORS (subsampled orthogonal with random sign) matrices. These are defined as matrices of the form $\mtx{A}=\mtx{HD}$, where $\mtx{H}$ is an $m\times d$ matrix arising from uniformly sampling $m$ rows from a unitary matrix with entries bounded by $\Delta/\sqrt{d}$ in absolute value for a constant $\Delta$, and $\mtx{D}$ is a diagonal matrix with a uniform random sign pattern on the diagonal. Using the results of~\cite{oymak2015isometric}, we get the following variation of Theorem~\ref{thm:mainA}.

\begin{theorem}\label{thm:oymak-main}
Under the conditions of Theorem~\ref{thm:mainA}, for a suitable constant $C$ and
\begin{equation*}
  m\geq C \cdot \Delta^2 (1+\log(1/\delta))^2\log^4(d) \frac{w^2(T)}{\e^2},
\end{equation*}
for a random $m\times d$ SORS matrix $\mtx{A}$,
with probability at least $1-\delta$, the persistence modules
associated to the \v{C}ech, Vietoris-Rips, and Delaunay complexes of $X$ and $\mtx{A}X$ are multiplicatively $(1-\e)^{-1}$-interleaved. 
\end{theorem}

As pointed out in~\cite{oymak2015isometric}, it is likely that the term $\log^4(d)$ can be reduced to $\log(d)$.
Important examples of SORS matrices with $\Delta=1$ are the (properly renormalized) Fourier transform, the discrete cosine transform, and the Hadamard transform, which allow for fast matrix-vector products. The possibility of computing
the dimension reduction efficiently is essential to the applicability of the reduction scheme. We revisit efficiency when discussing examples in Section~\ref{sec:experiments}.

\subsubsection{Applications} Two key advantages of the proposed dimension reduction scheme are that it is non-adaptive, and the fact that the target dimension of the projection does not depend on the size of the data set. Together with the possibility of using fast projections, the method has potential applications in settings where the data set changes with time, and one would like to update topological information as new data becomes available. More precisely, consider a given data set $X=\{\vct{x}_1,\dots,\vct{x}_k\}$, and assume that a new point $\vct{x}_{k+1}$ becomes available. The most basic operation, updating the distance matrix of the point set, requires $O(kd)$ operations. Assume that we have prior information on the type of data represented by $X$ (for example, that it consists of images that have a certain sparsity structure). If we store projections $\vct{P}\vct{x}_1,\dots,\vct{P}\vct{x}_{k}$, where $\mtx{P}$ has $m\ll d$ rows, then updating the distance matrix reduces to $O(km)$ operations after computing $\mtx{P}\vct{x}_{k+1}$. The cost of this reduction is the added complexity of computing the projection $\vct{P}\vct{x}_{k+1}$; when $\vct{P}$ is a sparse Johnson-Lindenstrauss transform with sparsity $s$, then the number of operations is $sd$ and the total cost of updating the distance matrix is $O(km+sd)$. When using a partial Fourier or Hadamard matrix, the cost becomes $O(km+d\log d)$. 

The setting most likely to benefit is when the dimension is large, the effective dimension (Gaussian width) is small, and the number of samples is sufficiently large. Note that when the number of samples $n$ is less than exponential in the Gaussian width, the bound~\eqref{eq:w-discrete}, namely $2\log(n)$, for the Gaussian width can be smaller than the bound implied by the underlying structure. For example, with $n$ points representing $s$-sparse signals in $\R^d$ we would require more than $2^sd\leq n$ samples for the cardinality-independent bound~\eqref{eq:w-sparse} 
to become more effective than~\eqref{eq:w-discrete}. 
We discuss some numerical examples relating the computation time to the achievable reduction in Section~\ref{sec:experiments}.
The benefits become more marked when dealing with more complex constructions.
For example, the size of a Delaunay triangulation can be of order $n^{d/2}$, as exemplified by the cyclic polytope~\cite{Ziegler1995},
and approximations such as the mesh filtration are of order $2^{O(d^2)}n$~\cite{hudson2010topological} (other approximations, such as sparse Rips filtrations~\cite{sheehy2013linear}, also achieve bounds linear in $n$ and exponential in $d$). Note that the complexity reduction can also play a role in the analysis of constructions that do not explicitly compute the projection.

\subsubsection{Limitations} If we are only interested in the Vietoris-Rips filtration, which depends only on pairwise distances, then Theorem~\ref{thm:mainA} extends to any metric that allows for low-distortion embeddings~\cite{indyk20048}, with appropriately adjusted bounds (see, for example, \cite{krahmer2016unified} for recent work on the $\ell_1$ norm). For \v{C}ech and Delaunay complexes, however, the statement depends crucially on Euclidean characterizations of mean and variance (Section~\ref{sec:sphere}) and is therefore restricted to Euclidean spaces (or data sets that can be embedded in such), and similarity measures derived from the Euclidean distance. As many applications of persistent homology involve metric spaces with non-Euclidean metrics, it would be interesting to see to what extent a practical randomized dimensionality reduction can be performed in this context. The results extend easily to weighted Euclidean distances, as in~\cite{sheehy2014persistent}.

\subsubsection{Robustness} With some modification, the results still apply in the presence of noise. In many practical settings the data points will satisfy structural constraints only approximately. For example,
images are generally not sparse but {\em compressible}, meaning that after some transform, all
but a few coefficients will be small but not exactly zero.  
Assume that that the data points are of the form $\vct{x}_i=\vct{y}_i+\vct{\nu}_i$, with $\norm{\vct{\nu}_i}\leq \nu$. If $\nu$ is large, then
the underlying structure is lost. In general, the squared Gaussian width can increase by a factor of up to $\nu^2d$, which restricts the method to small errors. Fortunately,
the coefficients of a trigonometric or Wavelet expansion of images are known to decay quickly, with the decay depending on the regularity properties of the image~\cite{mallat2008wavelet}.

\subsection{Related work}
The application of the Johnson-Lindenstrauss Theorem in relation to persistent homology was introduced by Sheehy~\cite{sheehy2014persistent}, on which our approach is based, and independently by Kerber and Raghvendra~\cite{kerber2014approximation}. In particular, a version of the key Theorem~\ref{thm:equivalence} with different constants appeared in~\cite{sheehy2014persistent}. These articles formulated their results using a target dimension of order $\log(n)/\e^2$, where $n$ is the cardinality of the point cloud. 
The idea of projecting to a space of dimension proportional to an intrinsic dimension, based on~\cite{baraniuk2009random,clarkson2008tighter}, was mentioned in~\cite[Section 4]{sheehy2014persistent}.
The work~\cite{kerber2014approximation} extends the Johnson-Lindenstrauss Theorem to the setting of projective clustering, of which the smallest enclosing ball is a special case.
Recently, the results presented in our paper have been extended to the computation of persistent homology using the $k$-distance~\cite{arya2018persistent}, answering a question by Sheehy~\cite{sheehy2014persistent}. The $k$-distance of a point to a set of points is defined by averaging the distance to the $k$ nearest neighbours in that set, and is more robust to outliers.

The doubling dimension has been used as a measure of intrinsic dimension in topological data analysis, see~\cite[Chapter 5]{oudot2015persistence} and the references therein, and also~\cite{sheehy2013linear, choudhary2014local}.
To our knowledge, the relation of the Gaussian width to the doubling dimension of a metric space was first pointed out by Indyk and Naor~\cite{indyk2007nearest}, even though a close connection is apparent in work on suprema of Gaussian processes~\cite{talagrand2014upper}. We revisit this relation in our context, with matching upper and lower bounds, in Section~\ref{sec:doubling}.

The use of the Gaussian width in compressed sensing was pioneered by Rudelson and Vershynin~\cite{Rudelson2008}, and has, in combination with Gordon's inequality, come to play a prominent role as a dimension parameter in the development of the theory~\cite{Foucart2013}. The Gaussian width also plays an important role in the analysis of signal recovery by convex optimization, as shown by~\cite{Stojnic2009} and generalized in~\cite{CRPW:12}, and a variation of the Gaussian width for convex cones, the statistical dimension, determines the location of phase transitions for the success probability of such problems~\cite{edge}. As far as we are aware, the Gaussian width has not yet been studied in the context of persistent homology. We would also like to point out that this notion of width is different from the notion of {\em smallest directional width} used in computational geometry; see~\cite{agarwal2004approximating} for more on this concept.

There has been extensive work on complexity reduction across all other parts of the persistent homology pipeline. These include subsampling techniques~\cite{chazal2015subsampling}, approaches to reduce the complexity of a filtration~\cite{de2004topological,chazal2008towards,sheehy2013linear,oudot2015zigzag}, and ways to improve on the matrix reduction~\cite{boissonnat2014simplex,boissonnat2015compressed,chen2011persistent,bauer2014clear,mendoza2017parallel}. We refer to~\cite{oudot2015persistence,kerber2016persistent,otter2017roadmap} for an overview and further references. 

\subsection{Outline of contents}
In Section~\ref{sec:persistent} we review in some detail the necessary prerequisites from persistent homology. This section also presents the basic interleaving result of Sheehy that links the Johnson-Lindenstrauss Theorem to the interleaving distance of persistence modules. Section~\ref{sec:gordon} reviews the Johnson-Lindenstrauss Theorem in the version of Gaussian matrices and Gaussian width, which is based on Gordon's inequalities for the expected suprema of Gaussian processes. Section~\ref{sec:sphere} presents a version of the Johnson-Lindenstrauss Theorem for smallest enclosing balls, which slightly improves a corresponding result by Sheehy~\cite{sheehy2014persistent}. This section also outlines a new proof of this result based on Slepian's Lemma. A direct consequence is a proof of Theorem~\ref{thm:mainA}. 
Section~\ref{sec:doubling} relates the Gaussian width to another intrinsic dimensionality parameter, the doubling dimension. Section~\ref{sec:experiments} presents some basic numerical experiments that illustrate that the dimensionality reduction can work in practice, while Section~\ref{sec:conclusion} discusses some further directions.

\subsection{Notation and conventions} 
For a set $S\subset \R^d$, let $\encl(S)$ denote the smallest enclosing ball, $\vct{c}_S$ its center and $\rho(S)$ its radius. Denote by $\partial \encl(S)$ the boundary and by $\partial S=\partial \encl(S)\cap S$ the points of $S$ on the boundary. Except when otherwise stated, the notation $\log$ will refer to the natural logarithm.

\section{Overview of Persistent Homology}\label{sec:persistent}
In persistent homology one usually begins with data interpreted as a point cloud, associates to it a filtration of simplicial complexes, constructs a boundary matrix related to the simplicial filtration, and then computes the persistence barcodes from a matrix reduction. These barcodes, or equivalently, the persistence diagrams, provide a topological summary of the data. 
We briefly review the part of the theory that is relevant to our purposes. There are many excellent references for the theory presented here, of which we would like to arbitrarily highlight~\cite{carlsson2009topology,edelsbrunner2010computational,oudot2015persistence}, the last of these being a good reference for both the module-theoretic perspective and as a survey of modern techniques and applications. For an overview of state-of-the-art software and a wealth of applications, see~\cite{otter2017roadmap}.

\subsection{Simplicial complexes}
General references for the material in this section are~\cite{munkres1984elements,hatcher2002algebraic} or the relevant chapters in~\cite{edelsbrunner2010computational}.
A simplicial complex $K$ is a finite collection of sets $\sigma$ that is closed under the subset relation. The elements $\sigma\in K$ are called simplices, and a subset $\tau\subseteq \sigma$ (itself a simplex) is called a face of $\sigma$. The {\em dimension} of a simplex is $\dim \sigma = |\sigma|-1$ (in particular, $\dim \emptyset = -1$), and a simplex of dimension $p$ is called a $p$-simplex. We denote by $K_p$ the set of all $p$-simplices. A map of simplicial complexes $f\colon K\to L$ is a map $f\colon K_0\to L_0$ such that $f(\sigma)\in L$ for all $\sigma \in K$. A map $f\colon K\to L$ between simplicial complexes is an isomorphism if it is injective, and $\sigma\in K \Leftrightarrow f(\sigma)\in L$.
A {\em subcomplex} of $K$ is a subset $L\subseteq K$ that is itself a simplicial complex. The $p$-skeleton $\bigcup_{k\leq p}K_k$ of $K$ is the subcomplex consisting of all simplices of dimension at most $p$.
We can associate to each $p$-simplex in $K$ a geometric simplex $\sigma$ in some $\R^d$ (that is, the convex hull of $p+1$ affinely independent points) in such a way that the face relations remain valid and the intersection of two simplices is either empty or again a simplex. The union of these geometric simplices is a {\em geometric realization} $|K|$ of the simplicial complex. A map between simplicial complexes $K\to L$ gives rise to a continuous map $|K|\to |L|$ between topological spaces. An important result in algebraic topology states that isomorphic simplicial complexes give rise to homeomorphic realizations~\cite{munkres1984elements}. In particular, the homotopy type (loosely speaking, the class of shapes that a set can be continuously deformed into) of a simplicial complex is well defined as the homotopy type of a realization of the complex.

Given a set $S\subset \R^d$ and a set of subsets $\mathcal{S}=\{U_i\}_{i\in I}$ such that $S\subseteq \bigcup_{i\in I} U_i$, we define the {\em nerve} $\mathcal{N}(\mathcal{S})$ of $\mathcal{S}$ to be the simplicial complex on the set $I$ defined by 
\begin{equation*}
  \sigma \in \mathcal{N}(\mathcal{S}) \quad \Leftrightarrow \quad \bigcap_{i\in \sigma} U_i \neq \emptyset.
\end{equation*}
The Nerve Theorem~\cite[4.G]{hatcher2002algebraic} asserts that if $\mathcal{S}$ is a finite collection of open, convex sets in $\R^d$, then the nerve $\mathcal{N}(\mathcal{S})$ is homotopy equivalent to the union $\bigcup_{U\in \mathcal{S}} U$. 
The Nerve Theorem also holds for covers with closed balls in Euclidean space~\cite[Chapter 4.3]{oudot2015persistence}, the setting in which it will be used in our case.

\subsection{Homology}
We consider homology over the field $\FF_2$ with two elements.
A chain complex $C_p(K)$ is the $\FF_2$-vector space generated by the $p$-simplices of a simplicial complex $K$. The boundary map $\partial_p$ maps a $p$-simplex to the sum of its $(p-1)$-dimensional faces,

\begin{equation*}
  \partial_p \colon C_p(K)\to C_{p-1}(K), \quad \sigma \mapsto \sum_{\tau\subset \sigma\cap K_{p-1}} \tau.
\end{equation*}

The boundary maps satisfy the fundamental property that for $p\geq 0$, $\partial_p\circ \partial_{p+1}=0$ (here, we use the convention that $\partial_0=0$). If we set $Z_p(K)=\ker \partial_p$ (the set of {\em cycles}) and $B_p(K) = \im \partial_{p+1}$ (the set of {\em boundaries}), then the $p$-th homology vector space is defined as the quotient
\begin{equation*}
  H_p(K) = Z_p(K)/B_p(K).
\end{equation*}
The $p$-th Betti number is $\beta_p = \dim H_p(K)$. Homology is functorial, meaning that a map $f\colon K\to L$ induces a morphism $f_*\colon H_p(K)\to H_p(L)$, with the property that an isomorphism of complexes maps to an isomorphism of homology groups. 

\subsection{Filtrations} 
To capture the topology of the data at different scales, we need to consider sequences of topological spaces and simplicial complexes ordered by inclusion.
Such sequences, called filtrations, give rise to sequences of homology vector spaces with their corresponding induced maps. 

A finite set $X=\{\vct{x}_1,\dots,\vct{x}_n\}\subset \R^d$ induces a filtration $\mathcal{X}=\{X_\alpha\}_{\alpha\in \R}$ of topological spaces, where
\begin{equation*}
  X_\alpha = \bigcup_{1\leq i\leq n} B(\vct{x}_i,\alpha)
\end{equation*}
is the union of the (closed) balls of radius $\alpha$ around the points in $X$ with respect to the Euclidean distance.
One can associate a simplicial filtration to this topological by taking nerve of the collection $\{B(\vct{x}_i,\alpha)\}$. For each $\alpha$, the resulting simplicial complex is the \v{C}ech complex,
\begin{equation*}
  \check{C}_\alpha(X) = \{\sigma \subset [n] \mid \bigcap_{i\in \sigma} B(\vct{x}_i,\alpha)\neq \emptyset\},
\end{equation*}
and as $\alpha$ varies this leads to the \v{C}ech filtration.
Equivalently, a subset $S=\{\vct{x}_i \mid i\in \sigma\}\subseteq X$ gives rise to a simplex $\sigma\in \check{C}_\alpha (X)$ if and only if $S\subset B(\vct{x},\alpha)$ for some $\vct{x}\in \R^d$. We will sometimes omit the brackets and write $\check{C}_\alpha X$, and equally with the other complexes introduced below.

The {\em Delaunay} filtration, or $\alpha$-filtration, is the sequence of simplicial complexes $D_\alpha(X)$ consisting of simplices $\sigma\subset [n]$ such that there exists $\vct{x}\in \R^d$ with \vspace{-0.3cm}
\begin{itemize}
\item $S=\{\vct{x}_i \mid i\in \sigma\} \subset B(\vct{x},\alpha)$;
\item for all $\vct{p}\not\in S$, $\vct{p}\not\in B(\vct{x},\alpha)$.
\end{itemize}
\vspace{-0.3cm}
Clearly, $D_\alpha(X) \subseteq \check{C}_\alpha(X)$. The Delaunay filtration has the advantage that if the points in $X$ are in general position (meaning that no $d+2$ of them lie on the surface of a common sphere), then  the simplices have dimension at most $d$, whereas for the \v{C}ech complex they can have dimension up to $n$. On the other hand, while the complexity of constructing a \v{C}ech complex only depends on the size of the data set (or the distance matrix), constructing a Delaunay complex has complexity exponential in the ambient dimension $d$, which makes it practical only for small dimensions. Finally, one of the most common constructions is the {\em Vietoris-Rips} complex, $V_\alpha(X)$, 
where $\sigma\in V_\alpha(X)$ if and only if for all $\vct{x},\vct{y}\in \sigma$, $\norm{\vct{x}-\vct{y}}\leq 2\alpha$.

A filtration gives rise to a filtration in the homology vector spaces: if $\alpha\leq \alpha'$, then we get a homomorphism
$\iota_{\alpha}^{\alpha'} \colon H_p(X_{\alpha})\to H_p(X_{\alpha'})$.
The $p$-th {\em persistent homology} of a filtration is the induced sequence of homology vector spaces and linear maps, and the $p$-th persistence vector spaces are the images of these homomorphisms,
\begin{equation*}
  H_p^{\alpha,\alpha'}(\mathcal{X}) = \im \iota_{\alpha}^{\alpha'} = Z_p(X_{\alpha})/(B_p(X_{\alpha'})\cap Z_p(X_{\alpha})).
\end{equation*}
Similarly, one can consider the simplicial filtration of \v{C}ech complexes, $\check{C}_\alpha(X)$, and the associated sequence of homology vector spaces. Such sequence of homology vector spaces with the associated maps $\iota_{\alpha}^{\alpha'}$ are examples of {\em persistence modules}, see~\cite{oudot2015persistence} for a more comprehensive treatment. 

The Nerve Theorem guarantees that for each $\alpha$ the homology of the nerve complex is the same as the homology of the cover, but it does not automatically follow that the persistent homology of the filtration induced by the cover is the same as the persistent homology of the resulting filtered simplicial complex. That this is the case is guaranteed by the Persistent Nerve Lemma~\cite[Lemma 4.12]{oudot2015persistence}.

\begin{lemma}(Persistent Nerve Lemma)
Let $X=\{\vct{x}_1,\dots,\vct{x}_n\}$ and $X_\alpha$ the cover with closed balls of radius $\alpha$. Let $\check{C}_\alpha(X)$ be the corresponding nerve. Then there is an isomorphism of persistence homology modules $H_p$.
Specifically, for every $\alpha$ there are isomorphisms $\nu_{\alpha}$ such that the following diagram commutes: 
\begin{equation*}
\begin{tikzcd}
  H_p(X_\alpha) \arrow[d, "\nu_\alpha"'] \arrow[r, "\iota_{\alpha}^{\alpha'}"] & H_p(X_{\alpha'}) \arrow[d, "\nu_{\alpha'}"]\\
  H_p(\check{C}_{\alpha}X) \arrow[r, "\iota_\alpha^{\alpha'}"] & H_p(\check{C}_{\alpha'}X)
\end{tikzcd}
\end{equation*}
\end{lemma}
The collection of maps $\nu_{\alpha}$ constitute a {\em morphism} of persistence modules (in this case, an isomorphism).
The Delaunay complex can equally be related to the persistent homology of the topological filtration $\mathcal{X}$ by noting that the Delaunay complex is a nerve of a cover of the same space. In all the situations of interest to us in this paper, the homology groups are finite-dimensional, and there are finitely many indices $c_0<c_1<\cdots<c_m$, the {\em critical points}, such that $H_p(X_{\alpha})=H_p(X_{\alpha'})$ and $\iota_{\alpha}^{\alpha'}=\mathrm{id}$ for $\alpha,\alpha'\in [c_i,c_{i+1})$.
While the homology of the Vietoris-Rips complex is not as directly related to the topological filtration via the Nerve Lemma, it approximates the \v{C}ech filtration, $\check{C}_\alpha(X) \subseteq V_\alpha(X) \subseteq \check{C}_{\sqrt{2}\alpha}(X)$.

For a simplicial filtration $\mathcal{K}=\{K^{\alpha_i}\}$ such as the \v{C}ech filtration, one defines the persistent homology vector spaces $H_p^{i,j}(\mathcal{K})$ just as in the case of a topological filtration. A simplex $\sigma$ is born at time $\alpha_i$ if appears in $H_p(K^{\alpha_i})$ but is not the image of an element in $H_p(K^{\alpha_{i-1}})$. A simplex $\sigma$ dies 
at time $\alpha_i$ if $\iota_{\alpha_{i-1}}^{\alpha_i}(\sigma)=0$. This way each element in the filtration, which corresponds to a topological feature of the realisation of the simplicial complex, comes with an interval $[\alpha,\alpha')$ representing its lifetime, where $\alpha'$ may be $\infty$. The lifetimes of the various features are recorded in a two-dimensional {\em persistence diagram}, where each interval $[a,b)$ is represented by a point with coordinates $(a,b)$ with multiplicity (which equals $0$ if there is no element whose lifetime matches the interval). Alternatively, one can represent each interval occurring using {\em persistence barcodes}, which record the lifetime of each feature as an interval. 
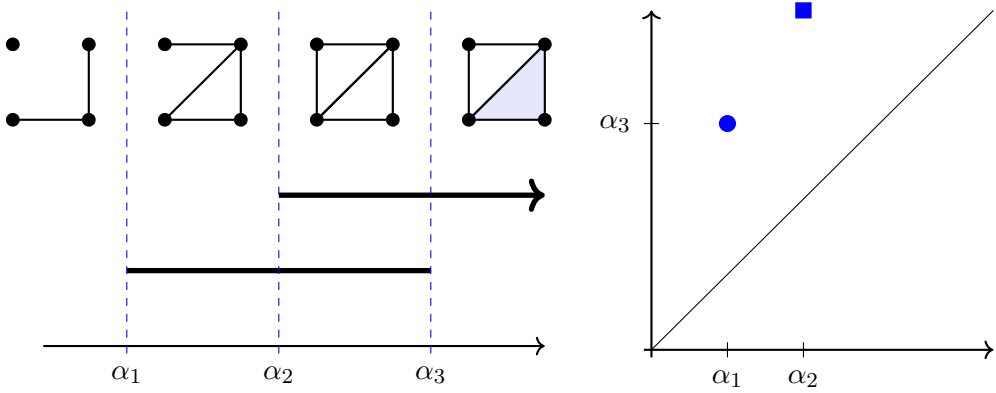
\begin{figure}[H]
\centering
\begin{minipage}{0.45\textwidth}
\begin{tikzpicture}[scale=0.6]
\draw[color=black, thick, ->] (-0.1,0)--(6.5,0);
\draw[color=black, line width=2pt] (1,1)--(5,1);
\draw[color=black, line width=2pt, ->] (3,2)--(6.5,2);
\draw[color=blue,dashed] (1,-0.1)--(1,4.5);
\draw[color=blue,dashed] (3,-0.1)--(3,4.5);
\draw[color=blue,dashed] (5,-0.1)--(5,4.5);
\node (A1) at (1,0)  [label=-90:{$\alpha_1$}] {};
\node (A2) at (3,0)  [label=-90:{$\alpha_2$}] {};
\node (A3) at (5,0)  [label=-90:{$\alpha_3$}] {};
  \draw[color=black, thick] (-0.5,3)--(0.5,3)--(0.5,4);
  \node (1) at (-0.5,3)   [circle, scale=0.5, fill=black] {};
  \node (2) at (-0.5,4)  [circle, scale=0.5, fill=black] {};
  \node (3) at (0.5,3)   [circle, scale=0.5, fill=black] {};
  \node (4) at (0.5,4)   [circle, scale=0.5, fill=black] {};
  
  \draw[color=black, thick] (1.5,3)--(2.5,3)--(2.5,4)--(1.5,3);
  \draw[color=black, thick] (1.5,4)--(2.5,4);
  \node (11) at (1.5,3)   [circle, scale=0.5, fill=black] {};
  \node (21) at (1.5,4)  [circle, scale=0.5, fill=black] {};
  \node (31) at (2.5,3)   [circle, scale=0.5, fill=black] {};
  \node (41) at (2.5,4)   [circle, scale=0.5, fill=black] {};
      
  \draw[color=black, thick] (3.5,3)--(3.5,4)--(4.5,4)--(3.5,3);
  \draw[color=black, thick] (3.5,3)--(4.5,3)--(4.5,4)--(3.5,3);
  \node (12) at (3.5,3)   [circle, scale=0.5, fill=black] {};
  \node (22) at (3.5,4)  [circle, scale=0.5, fill=black] {};
  \node (32) at (4.5,3)   [circle, scale=0.5, fill=black] {};
  \node (42) at (4.5,4)   [circle, scale=0.5, fill=black] {};
  
  \draw[color=black, thick] (5.5,3)--(5.5,4)--(6.5,4)--(5.5,3);
  \draw[color=black, fill=blue!10, thick] (5.5,3)--(6.5,3)--(6.5,4)--(5.5,3);
  \node (11) at (5.5,3)   [circle, scale=0.5, fill=black] {};
  \node (21) at (5.5,4)  [circle, scale=0.5, fill=black] {};
  \node (31) at (6.5,3)   [circle, scale=0.5, fill=black] {};
  \node (41) at (6.5,4)   [circle, scale=0.5, fill=black] {};
\end{tikzpicture}
\end{minipage}
\quad
 \begin{minipage}{0.45\textwidth}
\begin{tikzpicture}[scale=0.6]
\draw[color=black, thick, ->] (0,-0.1)--(0,4.5);
\draw[color=black, thick, ->] (-0.1,0)--(4.5,0);
\draw[color=black] (0,0)--(4.5,4.5);
\draw[color=black] (1,-0.1)--(1,0.1);
\draw[color=black] (2,-0.1)--(2,0.1);
\draw[color=black] (-0.1,3)--(0.1,3);
\filldraw[blue] (1,3) circle (3pt);
\filldraw[blue] (1.9,4.4) rectangle (2.1,4.6);
\node (A1) at (1,0)  [label=-90:{$\alpha_1$}] {};
\node (A2) at (0,3)  [label=-180:{$\alpha_3$}] {};
\node (A3) at (2,0)  [label=-90:{$\alpha_2$}] {};
\end{tikzpicture}
\end{minipage}
\caption{Persistence barcode and diagram for the first Betti number of a simplicial filtration. The circle represents a feature that is born at $\alpha_1$ and dies at $\alpha_3$, while the square represents a feature that is born at $\alpha_2$ and lives on.}
\end{figure}

\subsection{Stability}
In order to measure how changes in the input affect changes in the persistence diagrams, we need to define a notion of distance between persistence diagrams. Here we only discuss a notion of distance between persistence modules, for the relation to other common notions of distance of persistence diagrams we refer to~\cite{chazal2016structure}. In our setting, a {\em persistent module} is the sequence of homology vector spaces associated to a filtration with the maps induced by inclusion.

Let $X,Y$ be two point clouds with corresponding offset filtrations $\{X_\alpha\}$, $\{Y_\alpha\}$, and let $c>0$. 
A {\em multiplicative $c$-interleaving} is then given by a collection of morphisms $\varphi_{\alpha}\colon H_p(X_\alpha)\to H_p(Y_{c\alpha})$ and $\psi_{\alpha}\colon H_p(Y_\alpha)\to H_p(X_{c\alpha})$ such that for each $\alpha$, the following diagrams commute:

\begin{minipage}{0.5\textwidth}
\begin{equation*}
\begin{tikzcd}
  H_p(X_{\alpha/c}) \arrow[rd, "\phi_{\alpha/c}"'] \arrow[rr, "\iota_{\alpha/c}^{c\alpha}"] & & H_p(X_{c\alpha})\\
  & H_p(Y_{\alpha}) \arrow[ru, "\psi_{\alpha}"'] &
\end{tikzcd}
\end{equation*}
\end{minipage}
\begin{minipage}{0.5\textwidth}
\begin{equation*}
\begin{tikzcd}
  & H_p(X_{\alpha}) \arrow[rd, "\phi_{\alpha}"] & \\
  H_p(Y_{\alpha/c}) \arrow[ru, "\psi_{\alpha/c}"] \arrow[rr, "\iota_{\alpha/c}^{c\alpha}"] & & H_p(Y_{c\alpha})
\end{tikzcd}
\end{equation*}
\end{minipage}

The same definition applies to the persistence modules associated to a simplicial filtration. 
The (multiplicative) {\em interleaving distance} between two persistence modules is the smallest $c$ such that a multiplicative $c$-interleaving exists.
One similarly defines an additive interleaving by replacing $\alpha/c$ and $c\alpha$ with $\alpha-\e$ and $\alpha+\e$. A multiplicative interleaving is an additive interleaving on a logarithmic scale.

The following Lemma from~\cite{sheehy2014persistent} reduces the task of finding an interleaving of persistence modules to that of establishing inequalities for smallest enclosing balls. Recall the notation $\rho(S)$ for the radius of a smallest enclosing ball of $S$.

\begin{lemma}\label{lem:sheehy}
Let $X\subset \R^d$ be a finite set and $d_X\colon \R^d \to \R$ the distance function.
Assume we have a function $F\colon X\to \R^m$ such that for all subsets $S\subseteq X$ we have
\begin{equation}\label{eq:encl-ineq}
  (1-\e)\rho(S) \leq \rho(F(S)) \leq (1+\e)\rho(S).
\end{equation}
Then the persistent homology modules associated to the \v{C}ech and Delaunay filtrations of $X$ and $F(X)$ are multiplicatively $(1-\e)^{-1}$-interleaved.
\end{lemma}

\begin{proof}
We deal with the \v{C}ech-complex, the statement for the Delaunay filtration follows from a standard equivalence~\cite[Chapter 4]{oudot2015persistence}. To simplify notation,
set $c=(1-\e)^{-1}$. 

A set $S$ defines a simplex $\sigma\in \check{C}_\alpha(X)$ if and only if $\rho(S)\leq \alpha$. 
Let $S\subseteq X$ such that $\rho(S)\leq \alpha$, and let $\sigma\in \check{C}_\alpha(X)$ be the associated simplex in the \v{C}ech-complex. Then, by the second inequality in~\eqref{eq:encl-ineq},
\begin{equation*}
\rho(F(S))\leq (1+\e)\rho(S) \leq c\alpha,
\end{equation*}
so that $F$ induces a simplicial map $\check{C}_\alpha(X)\to \check{C}_{c\alpha}(F(X))$. 

Conversely, assume $S\subseteq F(X)$ gives rise to a simplex $\sigma\in \check{C}_{\alpha}(F(X))$. From the first inequality
of~\eqref{eq:encl-ineq} it follows that $F$ is injective, and therefore a bijection of $X$ to $F(X)$.
The inequality 
\begin{equation*}
(1-\e)\rho(S)\leq \rho(F^{-1}(S)) \leq \alpha,
\end{equation*}
then gives rise to a simplicial map $\check{C}_{\alpha}(F(X))\to \check{C}_{c\alpha}(X)$. Note that, by the injectivity of $F$, the composition $\check{C}_{\alpha/c}(X)\to \check{C}_{\alpha}(F(X))\to \check{C}_{c\alpha}(X)$ is the inclusion map and therefore gives rise to a multiplicative $c$-interleaving in the homology.
\end{proof}

\section{General Johnson-Lindenstrauss Transforms}\label{sec:gordon}
The classical Johnson-Lindenstrauss Theorem~\cite[15.2]{Matouvsek2002} shows the existence of a linear map $f\colon \R^d\to \R^m$, such that for all $\vct{x},\vct{y}$ from a finite set $X\subset \R^d$ with $|X|=n$,
\begin{equation*}
  (1-\e)\norm{\vct{x}-\vct{y}}\leq \norm{f(\vct{x})-f(\vct{y})} \leq (1+\e)\norm{\vct{x}-\vct{y}}, 
\end{equation*}
provided $m\geq C\cdot \log(n)/\e^2$ for some constant $C$.
 
This bound is sharp in general~\cite{larsen2014johnson}, but it can be refined based on a certain geometric measure of a set related to $X$. Arguably the most common geometric measure used in this context is the {\em Gaussian width} of a set $T$, defined as
\begin{equation*}
 w(T) = \Expect\sup_{\vct{x}\in T} \ip{\vct{g}}{\vct{x}},
\end{equation*}
where the expectation is taken over a random Gaussian vector in $\R^d$, i.e., $\vct{g}\in \normal(\zerovct,\onemtx)$. One version of the Johnson-Lindenstrauss Theorem can be stated as follows. In what follows we set $E_m := \Expect[\norm{\vct{g}}] = \sqrt{2}\Gamma((m+1)/2)/\Gamma(m/2)$, where $\Gamma(z)=\int_{0}^\infty x^{z-1}e^{-x}\ \diff{x}$ is the gamma function. It is known that
\begin{equation*}
  \frac{m}{\sqrt{m+1}} \leq E_m \leq \sqrt{m},
\end{equation*}
as follows, for example, from~\cite{wendel1948}.

\begin{theorem}\label{thm:j-l}(Johnson-Lindenstrauss - Gordon version)
Let $\delta \in (0,1)$, $X\subset \R^d$, and define $T=\{(\vct{x}-\vct{y})/\norm{\vct{x}-\vct{y}}_2 \mid \vct{x},\vct{y}\in X\}$. Assume that
\begin{equation*}
m \geq \frac{\left(w(T)+\sqrt{2\log(2/\delta)}\right)^2}{\e^2}+1.
\end{equation*} 
Then for a random Gaussian $m\times d$ matrix $\mtx{G}$, with entries $g_{ij}\sim N(0,1/E_m^2)$, we have
\begin{equation*}
  (1-\e)\norm{\vct{x}-\vct{y}} \leq \norm{\vct{G}\vct{x}-\vct{G}\vct{y}} \leq (1+\e)\norm{\vct{x}-\vct{y}},
\end{equation*}
uniformly for all $\vct{x},\vct{y}\in X$ with probability at least $1-\delta$.
\end{theorem}

As mentioned after Theorem~\ref{thm:mainA}, one can generalize Theorem~\ref{thm:j-l} with minor loss to subgaussian transformations~\cite{dirksen2016dimensionality}, to the setting of the Sparse Johnson Lindenstrauss Transform (SJLT)~\cite{bourgain2015toward} or more general so-called RIP-matrices \cite{oymak2015isometric, krahmer2011new}, that include, for example, partial Fourier or discrete cosine transforms. We present one such result,
which gives rise to Theorem~\ref{thm:oymak-main} in the same way as Theorem~\ref{thm:j-l} gives rise to Theorem~\ref{thm:mainA}. Recall that a SORS (subsampled orthogonal with ranodm sign) matrix is defined as a matrix of the form $\mtx{A}=\mtx{HD}$, where $\mtx{H}$ is an $m\times d$ matrix arising from uniformly sampling $m$ rows from a unitary matrix with entries bounded by $\Delta/\sqrt{d}$, and $\mtx{D}$ is a diagonal matrix with random i.i.d. sign pattern on the diagonal. 

\begin{theorem}(\cite[Theorem 3.3]{oymak2015isometric})\label{thm:oymak}
Let $\delta \in (0,1)$, $X\subset \R^d$, and define $T=\{(\vct{x}-\vct{y})/\norm{\vct{x}-\vct{y}}_2 \mid \vct{x},\vct{y}\in X\}$. Let 
$\mtx{A}\in \R^{m\times d}$ be a SORS matrix. Then for some constant $C$ and 
\begin{equation*}
  m\geq C \cdot \Delta^2 (1+\log(1/\delta))^2\log^4(d) \frac{w^2(T)}{\e^2},
\end{equation*}
the matrix $\mtx{A}$ satisfies
\begin{equation*}
  (1-\e)\norm{\vct{x}-\vct{y}} \leq \norm{\vct{A}\vct{x}-\vct{A}\vct{y}} \leq (1+\e)\norm{\vct{x}-\vct{y}},
\end{equation*}
uniformly for all $\vct{x},\vct{y}\in X$ with probability at least $1-\delta$.
\end{theorem}

Keeping in mind that the results presented here also hold in practically relevant settings, we nevertheless restrict the remaining discussion to the Gaussian case to keep the exposition conceptually simple.

The proof of Theorem~\ref{thm:j-l} is a well known and direct application of Theorem~\ref{thm:gordon}, which follows from an inequality of Gordon~\cite{Gordon1988} relating the expected suprema of Gaussian processes, together with concentration of measure for Lipschitz functions. We include the proof for convenience, an accessible derivation of Gordon's inequality itself can be found in the follow-up to~\cite[Theorem 9.21]{Foucart2013}.

\begin{theorem}(Gordon)\label{thm:gordon}
Let $\mtx{G}$ be an $m\times d$ matrix with standard Gaussian entries $g_{ij}\sim N(0,1)$, and let $T\subseteq S^{d-1}$ be a subset of the unit sphere. Then
\begin{align*}
  \Prob\{\min_{\vct{p}\in T} \norm{\mtx{G}\vct{p}}\leq E_m-w(T)-t\} &\leq e^{-t^2/2}\\
  \Prob\{\max_{\vct{p}\in T} \norm{\mtx{G}\vct{p}}\geq E_m+w(T)+t\} &\leq e^{-t^2/2}
  \end{align*}
\end{theorem}

\begin{proof}[Proof of Theorem~\ref{thm:j-l}]
Set $T=\{(\vct{x}-\vct{y})/\norm{\vct{x}-\vct{y}}_2 \mid \vct{x},\vct{y}\in X\}$. Then the claim is that with probability at least $1-\delta$, for all $\vct{p}\in T$,
\begin{equation*}
  E_m(1-\e)\leq  \norm{\tilde{\mtx{G}}\vct{p}}_2 \leq E_m(1+\e),
\end{equation*}
where we used the fact that $\mtx{G}=\frac{1}{E_m}\tilde{\mtx{G}}$ for a standard Gaussian matrix $\tilde{\mtx{G}}$. By the union bound, it suffices to show that
\begin{align*}
\begin{split}\label{eq:disp1}
 \Prob\{\min_{\vct{p}\in T} \norm{\tilde{\mtx{G}}\vct{p}}\leq E_m(1-\e)\} &\leq \frac{\delta}{2}\\
  \Prob\{\max_{\vct{p}\in T} \norm{\tilde{\mtx{G}}\vct{p}}\geq E_m(1+\e)\} &\leq \frac{\delta}{2}
  \end{split}
\end{align*}
This is where Gordon's Theorem~\ref{thm:gordon} comes into the picture. Set $t = \sqrt{2\log(2/\delta)}$,
so that $\delta/2 = e^{-t^2/2}$. The relation between $m$, $\e$ and $\delta$ in the statement of the theorem can be reformulated as
\begin{equation*}
w(T)+t\leq \e \sqrt{m-1} \leq \e \frac{m}{\sqrt{m+1}}\leq \e E_m,
\end{equation*}
and including the resulting bound into the inequalities in Theorem~\ref{thm:gordon} finishes the proof.
\end{proof}

\section{Enclosing Balls}\label{sec:sphere}
In view of Lemma~\ref{lem:sheehy}, what is needed is a version of the Johnson-Lindenstrauss Theorem involving smallest enclosing balls instead of pairwise distances. Such a result is provided by the following Theorem. It is a version of the Kirszbraun intersection property, which was shown for $|S|\leq d+1$ in~\cite[3.A]{gromov1987monotonicity}, and improves on a similar bound in~\cite{sheehy2014persistent}. 
We present a self-contained proof of Theorem~\ref{thm:equivalence} using elementary properties of the sample variance of a discrete distribution in $\R^d$, followed by a discussion of the relation to Slepian's inequality in Section~\ref{sec:kirszbraun}. Theorem~\ref{thm:equivalence} was derived independently by Sheehy~\cite{sheehy:pc}. Recall the notation $\rho(S)$ for the radius of the smallest enclosing ball of $S$.

\begin{theorem}\label{thm:equivalence}
Let $S\subset \R^d$ be a finite set and let $\e\in [0,1)$.  
Assume that for a map $f\colon \R^d\to \R^m$ and for all $\vct{x},\vct{y}\in S$ we have
\begin{equation}\label{eq:jlbound}
  (1-\e)\norm{\vct{x}-\vct{y}}\leq \norm{f(\vct{x})-f(\vct{y})}\leq (1+\e)\norm{\vct{x}-\vct{y}}.
\end{equation}
Then
\begin{equation}\label{eq:jl-balls}
 (1-\e)\rho(S)\leq \rho(f(S))\leq (1+\e)\rho(S).
\end{equation}
\end{theorem}

\begin{remark}
The literature on Johnson-Lindenstrauss is not always consistent on whether to use norms or squared norms, which leads to some ambiguity with respect to $\e$ and $\e^2$. We note that if we had used squared norms in the assumptions of Theorem~\ref{thm:equivalence}, we would also get the same result.
\end{remark}

\begin{remark}
For simplicity, Theorem~\ref{thm:equivalence} is stated for the Euclidean distance, but the proof of Theorem~\ref{thm:equivalence} can be extended to the case of the {\em power distance}~\cite{sheehy2014persistent,buchet2016efficient}. To define this distance, assign non-negative weights $w(\vct{x})$ to points in $X$ and $f(X)$. Then the power distance from $\vct{x}$ to a weighted point $\vct{y}$ is defined as
\begin{equation*}
  d_{\vct{y}}(\vct{x})^2 = \norm{\vct{x}-\vct{y}}^2+w(\vct{y})^2. 
\end{equation*}
The radius of a smallest enclosing ball is then defined as
\begin{equation*}
  \rho_w(S) = \min_{\vct{p}\in \R^d} \max_{\vct{x}\in X} d_{\vct{x}}(\vct{p}),
\end{equation*}
with the minimizing $\vct{p}$ as center $\vct{c}_S$. One easily checks that Lemma~\ref{le:conv} extends to this setting, and
assuming that $w(\vct{x})=w(f(\vct{x}))$, the proof of Theorem~\ref{thm:equivalence} carries over.
\end{remark}

To prepare for the proof of Theorem~\ref{thm:equivalence} we first need a few elementary auxiliary results. Lemma~\ref{le:conv} appears to be folklore.
For a set $S$, the {\em center} $\vct{c}_S$ of $S$ is the center of the smallest including ball.

\begin{lemma}\label{le:conv}
Let $S$ be a set and let $\vct{c}_S$ denote the center of $S$. Then $\vct{c}_S\in \conv(S)$. 
\end{lemma}

\begin{proof}
Assume $\vct{c}_S\not\in \conv(S)$ and 
denote by $\Proj(\vct{c}_S)=\argmin_{\vct{x}\in \conv(S)} \norm{\vct{c}_S-\vct{x}}$ the projection of $\vct{c}_S$ onto $\conv(S)$. We show that any point in $S$ is closer to $\Proj(\vct{c}_S)$ than to $\vct{c}_S$. In fact, for any $\vct{p}\in S$ we get
\begin{align*}
  \norm{\vct{c}_S-\vct{p}}^2 &= \norm{\vct{c}_S-\Proj(\vct{c}_S)+\Proj(\vct{c}_S)-\vct{p}}^2\\
  & = \norm{\vct{c}_P-\Proj(\vct{c}_S)}^2+\norm{\Proj(\vct{c}_S)-\vct{p}}^2-2\ip{\vct{c}_S-\Proj(\vct{c}_S)}{\vct{p}-\Proj(\vct{c}_S)}\\
  & \geq \norm{\Proj(\vct{c}_S)-\vct{p}}^2,
\end{align*}
where for the inequality we used the separating hyperplane theorem, which implies that the inner product in the expression is non-positive.
\end{proof}

The following elementary observation is just the expression of the sample variance of a point set in terms of pairwise distances.

\begin{lemma}\label{le:ball}
Let $\vct{c}=\sum_{i=1}^k \lambda_i \vct{x}_i$ be a convex combination of elements of $S$. Then
\begin{equation}\label{eq:firstconvex}
 \sum_{i=1}^k \lambda_i \norm{\vct{x}_i-\vct{c}}^2 = \sum_{i<j} \lambda_i\lambda_j \norm{\vct{x}_j-\vct{x}_i}^2.
\end{equation}
\end{lemma}

\begin{proof}
Using the representation of $\vct{c}$ as convex combination of the $\vct{x}_i$, we get
\begin{equation}\label{eq:first}
 \norm{\vct{x}_j-\vct{c}}^2 = \ip{\vct{x}_j-\vct{c}}{\vct{x}_j-\sum_{i=1}^k \lambda_i\vct{x}_i} = \sum_{i=1}^k \lambda_i\ip{\vct{x}_j-\vct{c}}{\vct{x}_j-\vct{x}_i}.
\end{equation}
Each summand can be characterized as
\begin{equation*}
  \ip{\vct{x}_j-\vct{c}}{\vct{x}_j-\vct{x}_i} = \frac{1}{2}\left(\norm{\vct{x}_j-\vct{c}}^2+\norm{\vct{x}_j-\vct{x}_i}^2-\norm{\vct{x}_i-\vct{c}}^2\right).
\end{equation*}
Plugging this identity into~\eqref{eq:first}, using $\sum_{i=1}^k\lambda_i=1$, and combining all terms involving $\norm{\vct{x}_j-\vct{c}}^2$,
\begin{equation*}
\norm{\vct{x}_j-\vct{c}}^2 = \sum_{i=1}^k \lambda_i\norm{\vct{x}_j-\vct{x}_i}^2-\sum_{i=1}^k \lambda_i \norm{\vct{x}_i-\vct{c}}^2.
\end{equation*} 
Rearranging and summing both sides with weights $\lambda_j$ establishes the claim.
\end{proof}

\begin{proof}[Proof of Theorem~\ref{thm:equivalence}]
Let $\partial S = \{\vct{x}_1,\dots,\vct{x}_k\}$ and let $\vct{c}_S=\sum_{i=1}^k \lambda_i \vct{x}_i$ be representation of the center of $S$ as a convex combination of the elements of $\partial S$. Set $\tilde{\vct{c}}=\sum_{i=1}^k \lambda_i f(\vct{x}_i)$. Applying Lemma~\ref{le:ball} twice, we get
\begin{align*}
\rho^2(S) = \sum_{i=1}^k \lambda_i\norm{\vct{x}_i-\vct{c}_S}^2 &\stackrel{\eqref{eq:firstconvex}}{=} \sum_{i<j}\lambda_i\lambda_j \norm{\vct{x}_j-\vct{x}_i}^2\\
&\leq \frac{1}{(1-\e)^2}\left(\sum_{i<j}\lambda_i\lambda_j \norm{f(\vct{x}_j)-f(\vct{x}_i)}^2\right)\\
&\stackrel{\eqref{eq:firstconvex}}{=} \frac{1}{(1-\e)^2}\left(\sum_{i=1}^k \lambda_i\norm{f(\vct{x}_i)-\tilde{\vct{c}}}^2\right).
\end{align*}
As the function $\vct{c}\mapsto \sum_{i=1}^k \lambda_i\norm{f(\vct{x}_i)-\vct{c}}^2$
is minimized at $\tilde{\vct{c}}$, we can continue the above and conclude
\begin{equation*}
\rho^2(S)\leq \frac{1}{(1-\e)^2}\left(\sum_{i=1}^k \lambda_i\norm{f(\vct{x}_i)-\vct{c}_{f(S)}}^2\right)
\leq \frac{1}{(1-\e)^2}\rho^2(f(S)).
\end{equation*}
For the right-hand inequality we proceed similarly. 
\end{proof}

\begin{proof}[Proof of Theorem~\ref{thm:mainA}]
The Johnson-Lindenstrauss Theorem~\ref{thm:j-l} states that under the assumptions of Theorem~\ref{thm:mainA}, the pairwise distances are preserved up to multiplicative factors of $1\pm \e$. By Theorem~\ref{thm:equivalence}, the smallest enclosing balls are preserved up to the same factors. Finally, Lemma~\ref{lem:sheehy} asserts that this
gives rise to the desired interleaving of persistent modules. This completes the proof.
\end{proof}

\subsection{Slepian's Lemma and the Kirszbraun intersection property}\label{sec:kirszbraun}
The Kirszbraun intersection property~\cite[3.A]{gromov1987monotonicity} states that, given sets of distinct points $X=\{\vct{x}_1,\dots,\vct{x}_k\}$ and $Y=\{\vct{y}_1,\dots,\vct{y}_k\}$ in $\R^d$ with $k\leq d+1$ and 
\begin{equation}\label{eq:pairwise-ineq}
  \norm{\vct{x}_i-\vct{x}_j} \geq \norm{\vct{y}_i-\vct{y}_j}
\end{equation}
for $1\leq i<j\leq k$, then 
\begin{equation}\label{eq:intersection}
  \bigcap_{\vct{x}\in X} B_r(\vct{x}) \neq \emptyset \Rightarrow \bigcap_{\vct{y}\in Y} B_{r}(\vct{y}) \neq \emptyset,
\end{equation}
where $B_r(\vct{x})$ denotes the closed ball of radius $r$ around $\vct{x}$.
This, in turn, is equivalent to $\rho(Y)\leq \rho(X)$, and applied to $\vct{y}_i=f(\vct{x}_i)/(1+\e)$ implies the upper bound in Theorem~\ref{thm:equivalence}. The lower bound follows similarly.
The proof of Theorem~\ref{thm:equivalence} can therefore be seen as an alternative (and more direct) derivation of this intersection property without any restriction on $k$. The Kirszbraun intersection property has been used in~\cite{bauer2017dynamical} to study sampled dynamical systems.

The proof of the Kirszbraun intersection property given in~\cite{gromov1987monotonicity} is of independent interest, as it is based on the intuitive observation that the volume of the intersection of balls centered at a finite set of points increases as the points move closer together (see also~\cite{gordon1995volume} for variations on this theme). This observation also suggests a connection to Slepian's Lemma from the theory of Gaussian processes~\cite{Ledoux1991}. In fact, Theorem~\ref{thm:equivalence} can be derived from Slepian's Lemma, as we show next.
One version of Slepian's Lemma,
as stated in~\cite{M:11} and~\cite[Appendix B]{effective}, is as follows.

\begin{theorem}[Slepian Inequality]\label{thm:slepian}
Let $X_1,\dots,X_k$ and $Y_1,\dots,Y_k$ be centered Gaussian random variables such that
\begin{equation*}
  \Expect[(X_i-X_j)^2] \geq \Expect[(Y_i-Y_j)^2]
\end{equation*}
for $1\leq i<j\leq k$. Let $f\colon \R\to \R$ be a monotonically non-decreasing function. Then
\begin{equation*}
  \Expect[\max_{1\leq i\leq k} f(X_i)] \geq \Expect[\max_{1\leq i\leq k} f(Y_i)].
\end{equation*}
\end{theorem}

\begin{proof}[Proof of Theorem~\ref{thm:equivalence}, Slepian version]
Let $X=\{\vct{x}_1,\dots,\vct{x}_k\}$ and $Y=\{\vct{y}_1,\dots,\vct{y}_k\}$ be sets of distinct points in $\R^d$ satisfying the inequalities~\eqref{eq:pairwise-ineq}. It is enough to show that $\rho(X)\geq \rho(Y)$. 
Assume without loss of generality that the smallest enclosing balls of $X$ and $Y$ are centered at $\zerovct$, and set $R=\rho(X)$, $R'=\rho(Y)$. Define the random variables $X_i=\ip{\vct{x}_i}{\vct{g}}$ and $Y_i=\ip{\vct{y}_i}{\vct{g}}$ for $1\leq i\leq k$, where $\vct{g}$ is a centered Gaussian random vector in $\R^d$. Then for $1\leq i<j\leq k$,
\begin{equation*}
  \Expect[(X_i-X_j)^2] = \Expect[\ip{\vct{x}_i-\vct{x}_j}{\vct{g}}^2] = \norm{\vct{x}_i-\vct{x}_j}^2 \geq  \norm{\vct{y}_i-\vct{y}_j}^2 = \Expect[(Y_i-Y_j)^2],
\end{equation*}
so that the conditions of Theorem~\ref{thm:slepian} are satisfied.

\begin{figure}[H]
\centering
\begin{tikzpicture}[thick,rotate=0,scale=0.5]
\coordinate (P1) at (-2,-1);
\coordinate (P2) at (2,-1);
\coordinate (P3) at (-1,2);
\coordinate (C) at (0,0);
\coordinate (RP1) at (2,1);
\coordinate (RP2) at (-2,1);
\coordinate (RP3) at (1,-2);

\draw[color=black, thick] (C) -- (P1);
\draw[color=black, thick] (C) -- (P2);
\draw[color=black, thick] (C) -- (P3);
\filldraw[black] (P1) circle (2pt);
\filldraw[black] (P2) circle (2pt);
\filldraw[black] (P3) circle (2pt);
\filldraw[black] (C) circle (2pt);
\node (N1) at (P1) [label=180:$\vct{x}_1$] {};
\node (N2) at (P2) [label=0:{$\vct{x}_2$}] {};
\node (N3) at (P3) [label=90:{$\vct{x}_3$}] {};
\node (NC) at (C) [label={90:{$\quad 0$}}] {};
\node (N1) at (P1) [label=180:$\vct{x}_1$] {};
\node (N2) at (P2) [label=0:{$\vct{x}_2$}] {};
\node (N3) at (P3) [label=90:{$\vct{x}_3$}] {};

\def\firstcircle {(P1) circle (3cm)}
\def\secondcircle {(P2) circle (3cm)}
\def\thirdcircle {(P3) circle (3cm)}
\draw[black, thick] \firstcircle;
\draw[black, thick] \secondcircle;
\draw[black, thick] \thirdcircle;

\draw[black, thick, dashed] (C) circle (0.745cm);

\draw[color=black, thick, dashed] (C) -- (RP1);
\draw[color=black, thick, dashed] (C) -- (RP2);
\draw[color=black, thick, dashed] (C) -- (RP3);

\coordinate (T1) at (0.683,0.342);
\coordinate (T2) at (-0.683,0.342);
\coordinate (T3) at (0.342,-0.683);
\draw[color=blue, thick] (0.683-1,0.342+2) -- (0.683+1,0.342-2);
\draw[color=blue, thick] (-0.683+1,0.342+2) -- (-0.683-1.5,0.342-3);
\draw[color=blue, thick] (0.342-3,-0.683-1.5) -- (0.342+2,-0.683+1);

\begin{pgfonlayer}{bg}
 \draw[black!0, fill=red!25, thick] (C) circle (3cm);
 \draw[black!0, fill=green!25, thick] (-1.708,-1.708) -- (0,1.708) -- (1.04,-0.35) -- cycle;
 \clip \firstcircle;
 \clip \secondcircle;
 \fill[blue!25] \thirdcircle;
\end{pgfonlayer}

\end{tikzpicture}
\caption{Slepian's Lemma and enclosing balls. The triangle represents the polar body whose Gaussian measure is the subject of Slepian's Inequality.}\label{fig:circles}
\end{figure}
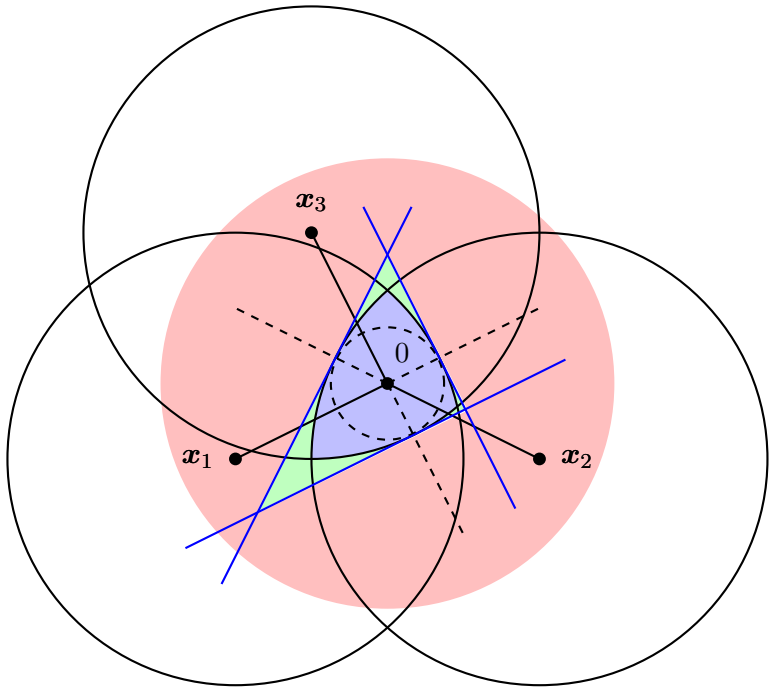

If we apply Theorem~\ref{thm:slepian} with $f(x)=x$,
we get a comparison between the Gaussian widths of the two sets.
We can therefore apply Theorem~\ref{thm:slepian} with $f(x)=\onemtx\{x> t\}$ to conclude that, by using the complements,
\begin{equation}\label{eq:slepian-ineq}
  \Prob\{\max_i X_i \leq t\} \leq \Prob\{\max_i Y_i \leq t\}.
\end{equation}
for all $t>0$. If we set $P=\conv\{\vct{x}_1,\dots,\vct{x}_k\}$, then the set $\{\vct{x}\mid \max_i \ip{\vct{x}_i}{\vct{x}}\leq t\}$ is just
$t P^*$, where $P^*$ is the polar body of $P$ (see Figure~\ref{fig:circles}), and $\Prob\{\max_i X_i \leq t\}$ is the Gaussian measure of this polar body.
If this polar body is not empty, it contains a closed ball of radius $t/R$: in fact, if $\vct{x}$ is in such a ball, then for all $i$, $\ip{\vct{x}_i}{\vct{x}}\leq \norm{\vct{x}_i}\norm{\vct{x}} \leq t$. By the same reasoning, the set $\{\vct{x}\mid \max_i\ip{\vct{y}_i}{\vct{x}}\leq t\}$ contains a ball of radius $t/R'$. We can therefore decompose the probability
\begin{equation*}
  \Prob\{\max_i X_i \leq t\} = \Prob\left\{\norm{\vct{g}} \leq \frac{t}{R}\right\} + g(t), \quad \Prob\{\max_i Y_i \leq t\} = \Prob\left\{\norm{\vct{g}} \leq \frac{t}{R'}\right\} + h(t),
\end{equation*}
with $g(t),h(t)\to 0$ as $t\to \infty$ (this follows from the decay of the Gaussian distribution). In particuliar,~\eqref{eq:slepian-ineq} implies that for every $\e>0$ there
exists a $t_0>0$ such that
\begin{equation*}
  \Prob\left\{\norm{\vct{g}} \leq \frac{t}{R}\right\} \leq \Prob\left\{\norm{\vct{g}} \leq \frac{t}{R'}\right\} + \e
\end{equation*}
for all $t>t_0$. This is only possible if $R\geq R'$, which concludes the proof.
\end{proof}

\section{Gaussian width, entropy, and doubling dimension}\label{sec:doubling}
A common notion in the study of metric spaces is the {\em doubling dimension}~\cite{assouad1983plongements, heinonen2012lectures}. This concept has been used to measure the intrinsic dimension of sets in topological data analysis, see~\cite[Chapter 5]{oudot2015persistence} or~\cite{sheehy2013linear,choudhary2014local} for some examples.

\begin{definition}
The {\em doubling constant} of a metric space $X$ is the smallest number $\lambda$ such that every ball of radius $R$ can be covered by $\lambda$ balls of radius $R/2$. The {\em doubling dimension} is defined as $\dim_d(X)=\log_2(\lambda)$.
\end{definition}

It is intuitively clear that the doubling dimension of Euclidean space $\R^d$ is of order $d$, but it can be considerably lower for certain subsets of $\R^d$. The metric spaces we consider here are finite subspaces $X\subset \R^d$. The {\em diameter} $\diam(X)$ is the largest pairwise distance between points in $X$. The {\em spread} $\Delta$ of such a set is the ratio of the diameter to the smallest pairwise distance between points in $X$. If a space $X$ has doubling dimension $\dim_d(X)$, then it is easy to see that any ball of radius $R$ can be covered with $\lceil\lambda^{\log_2(R/r)}\rceil = \lceil (R/r)^{\dim_d(X)}\rceil$ balls of radius $r$. 
From this it can be deduced that the cardinality $n=|X|$, the spread $\Delta$ and the doubling dimension are related as
\begin{equation}\label{eq:ndeltad}
 n\leq \Delta^{\dim_d(X)}.
\end{equation}
In~\cite{sheehy2013linear}, a linear size approximation to the Vietoris-Rips complex has been derived and analysed in terms of the doubling dimension.
The approach was further extended in~\cite{choudhary2014local}, where a local version of the doubling dimension is introduced.

It turns out that the doubling dimension is closely related to the Gaussian width, a fact pointed out in~\cite{indyk2007nearest}. This relationship provides an alternative way of expressing the cardinality of a set of points in terms of the spread and some intrinsic geometric parameter. To make this relationship precise, we need to introduce the inequalities of Dudley and Sudakov. References are \cite{talagrand2014upper} or \cite[Chapter 13]{boucheron2013concentration}.

A subset $\mathcal{N}_\alpha\subset X$ is called an $\alpha$-net, if $\norm{\vct{x}-\vct{y}}>\alpha$ for distinct $\vct{x},\vct{y}\in \mathcal{N}_\alpha$, and $\mathcal{N}_\alpha$ has maximal cardinality among sets with this property. 
For $\alpha>0$ let $N(X,\alpha)$ denote the cardinality of an $\alpha$-net. The logarithm $H(X,\alpha) = \log N(X,\alpha)$ is often referred to as the {\em metric entropy} of $X$. (One version of) Dudley's upper bound on the Gaussian width in terms of metric entropy is given as follows (following~\cite[Corollary 13.2]{boucheron2013concentration}).

\begin{theorem}(Dudley~\cite{dudley1967sizes})
Let $X\subset \R^d$ be a finite set. Then
\begin{equation}\label{eq:dudley}
 w(X) \leq 12\int_0^{\diam(X)/2} \sqrt{H(X,t)} \ \diff{t}.
\end{equation}
\end{theorem}

There is a corresponding lower bound, due to Sudakov. A reference is again~\cite{boucheron2013concentration}, though explicit constants are never included in the literature.

\begin{theorem}(Sudakov~\cite{sudakov})
Let $X\subset \R^d$ be a finite set with $r$ the smallest distance between points in $X$. Then
\begin{equation}\label{eq:sudakov}
 w(X) \geq \frac{3}{5}r \sqrt{\log(|X|)}.
\end{equation}
\end{theorem}

The upper bound in the following Proposition is from~\cite[(2)]{indyk2007nearest}, with the difference that, as in Theorem~\ref{thm:mainA}, we look at the Gaussian width of the set of normalised differences,
\begin{equation*}
 T = \left\{\frac{\vct{x}-\vct{y}}{\norm{\vct{x}-\vct{y}}} \mid \vct{x},\vct{y}\in X\right\}.
\end{equation*}
Contrary to the tradition, the bounds are stated using rather specific constants instead of only ``some universal constant $C$ or $L$''. While the precise values depend on details of the chosen analysis and are not important, for someone looking into actually using dimensionality reduction schemes it may be of interest to know if the ``universal constants'' are in the tens or in the billions.

\begin{proposition}\label{prop:gwidth-dim}
Let $X\subset \R^d$ be a finite set. Then
\begin{equation*}
  \frac{36}{25}\cdot \Delta(X)^{-2} \dim_d(X) \leq w^2(T) \leq 227 \cdot \Delta(X)^2 \dim_d(X)
\end{equation*}
\end{proposition}

\begin{proof}
We first relate the Gaussian width of $T$ to that of $X$. 
Let $R=\diam(X) = \max_{\vct{x},\vct{y}}\norm{\vct{x}-\vct{y}}$ and $r=\min_{\vct{x}\neq \vct{y}}\norm{\vct{x}-\vct{y}}$, so that $\Delta = R/r$. Note that, by the symmetry of the Gaussian distribution,
\begin{equation*}
 w(X-X) = \Expect_{\vct{g}}\sup_{\vct{x},\vct{y}\in X} \ip{\vct{g}}{\vct{x}-\vct{y}} = 2\Expect\sup_{\vct{x}\in X}\ip{\vct{g}}{\vct{x}} = 2w(X).
\end{equation*}
Using this, one readily derives the bounds
\begin{equation}\label{eq:diffwidths}
  \frac{2}{R}w(X) \leq w(T) \leq \frac{2}{r} w(X).
\end{equation}

The upper bound in the statement of the theorem was given in~\cite{indyk2007nearest}, though we recreate the argument with slightly better constants.
From~\eqref{eq:ndeltad}, we get the inequality
\begin{equation*}
  N(X,\alpha) \leq \left(\frac{R}{\alpha}\right)^{\dim_d(X)},
\end{equation*}
which implies, using Dudley's bound~\eqref{eq:dudley},
\begin{align*}
  w(X)&\leq 12 \sqrt{\dim_d(X)} \int_0^{R/2} \sqrt{\log(R/t)} \ \diff{t} \\
  &= 12 R\sqrt{\dim_d(X)} \int_0^{1/2} \sqrt{\log(1/t)} \ \diff{t}\\
  & \leq 12 R\sqrt{\dim_d(X)} \int_0^1 \sqrt{\log(1/t)} \ \diff{t} = 12 R \sqrt{\dim_d(X)} \frac{\sqrt{\pi}}{2}.
\end{align*}
Squaring the right-hand side and combining with~\eqref{eq:diffwidths} gives the desired bound.

For the lower bound, let $\vct{p}\in X$ and $B(\vct{p},\alpha)$ be a ball of radius $\alpha$ such that there is a minimal covering $\mathcal{C}_{\alpha/2}$ of $B(\vct{p},\alpha)\cap X$ with $\lambda$ balls $B(\vct{x}_i,\alpha/2)$, $1\leq i\leq \lambda$, where $\lambda$ is the doubling constant. Since the covering is minimal, we have $\norm{\vct{x}_i-\vct{x}_j}\geq \alpha$ for all $1\leq i,j\leq \lambda$. Set $S=\{\vct{x}_1,\dots,\vct{x}_\lambda\}$. Then
\begin{equation*}
 w(X) \geq w(S) \geq \frac{3}{5}\alpha \sqrt{\log_2(|S|)} \geq \frac{3}{5} r \sqrt{\dim_d(X)},
\end{equation*} 
where the first inequality follows from the monotonicity of the Gaussian width, and the second from Sudakov's inequality. Squaring the right-hand side and combining with~\eqref{eq:diffwidths}  shows the lower bound.
\end{proof}

The ``correct'' way of bounding the Gaussian width from above and below would be via Talagrand's $\gamma_2$ functional. Following~\cite[Section 2.2]{talagrand2014upper}, we call a nested sequence of partitions $\mathcal{A}=(\mathcal{A}_n)$ of $X$ {\em admissible}, if each partition satisfies the cardinality bound $|\mathcal{A}_n|\leq 2^{2^n}$ for $n\geq 1$. For each $n$ and each $\vct{x}\in X$ there is a unique element $A_n(\vct{x})\in \mathcal{A}_n$ which contains $\vct{x}$. The $\gamma_2$ functional is then defined as 
\begin{equation*}
 \gamma_2(X) = \inf_{\mathcal{A}} \sup_{\vct{x}\in X} \sum_{n\geq 0} 2^{n/2} \diam(A_n(\vct{x})),
\end{equation*}
where the infimum is over all admissible partition sequences.
Talagrand's Majorizing Measures Theorem~\cite[Theorem 2.4.1]{talagrand2014upper} gives bounds on the Gaussian width in terms of this functional.
For some constant $C$, 
\begin{equation*}
 \frac{1}{C} \gamma_2(X) \leq w(X) \leq C \gamma_2(X).
\end{equation*}
Note that we can alternatively represent an admissible sequence of partitions as a hierarchical tree, where each level approximates the data set more accurately. 
We suspect that proofs based on net-trees, such as those of the statements in~\cite{sheehy2013linear,choudhary2014local}, could be formulated in terms of the Gaussian width by associating to the point cloud (and all the associated nets) a stochastic process $g_{\vct{x}} = \ip{\vct{g}}{\vct{x}}$ for $\vct{x}\in T$.
 
\section{Experiments}\label{sec:experiments}
The theory behind the bounds such as the one in Theorem~\ref{thm:oymak-main} is involved, and the bounds are, in the form stated, primarily of theoretical interest and more explanatory than practical.
To evaluate the applicability of the random projection method, experiments are needed. A first aim is to establish whether the added cost of projecting the data to a lower dimensional space offsets the benefits of working in lower dimensions. 

The minimum amount of work on the data set consists of computing the pairwise distances of the data points points, so we first focus on this task. If $f(d)$ is the cost of projecting one data vector and $c(d)$ the cost of computing the distance between two vectors in $\R^d$, then the total cost of computing all the distances is $n(n-1)c(d)/2$ for the original data, and $n(n-1)c(m)/2+nf(d)$ when applying projection. It follows that the projection is effective whenever the number of samples satisfies
$n>2 f(d)/(c(d)-c(m))+1$.
When measuring the cost in number of arithmetic operations, then typically $f(d) = C\cdot d\log d$ (when using FFT-based algorithms) or $f(d) = sd$ (when using a sparse Johnson-Lindenstrauss transform), and $c(d)-c(m) = 3(d-m)$. If we fix a proportion $m=\sigma d$, then the minimum number of samples for the projection to be computationally effective is proportional to $\log(d)$ (when using FFT-based methods) or $s$ (when using sparse projections). Note that the theory requires random sign changes before projecting, adding $d$ operations (these are cheaper in practice than a normal arithmetic operation). When using approximations that rely on a subset of the distances of size linear in $n$, then we get an improvement if the projection can be computed faster than a certain multiple of $c(d)-c(m)$, independent of the number of samples.

In order to empirically test the projection costs, we use an example using a subsampled fast Hadamard transform. The Hadamard transform $H_m$
is a $2^m\times 2^m$ matrix, defined recursively as
\begin{equation*}
 H_m = \frac{1}{\sqrt{2}}\begin{pmatrix}
   H_{m-1} & H_{m-1}\\
   H_{m-1} & -H_{m-1}
 \end{pmatrix},
\end{equation*}
with $H_0=1$. It is an orthogonal transformation that can be multiplied efficiently to a vector in $\R^d$ with $O(d\log d)$ operations using a variant of the Fast Fourier Transform.
It is also an example of a Bounded Orthonormal System, and randomly subsampled rows of this matrix satisfy the Restricted Isometry Property of order $s$ with good constants,
provided $m\geq C s\log(d)$~\cite[Chapter 12.1]{Foucart2013}. In particular, it is a transform that can be used in conjuction with Theorem~\ref{thm:oymak-main}.
For computing the Fast Hadamard Transform (FHT), we used the FFHT Python package, which is part of the FALCONN project~\cite{FALCONN}. 
Figure~\ref{fig:savings} shows the number of samples at which the dimensionality reduction leads to computational savings when assembling the distance matrix (that is, when the cost of projecting and then computing all distances is less than the cost of simply computing all distances), for two examples of ambient dimension ($d=256$ and $d=4096$). 
\begin{figure}[H]
\centering
\includegraphics[width=0.8\textwidth]{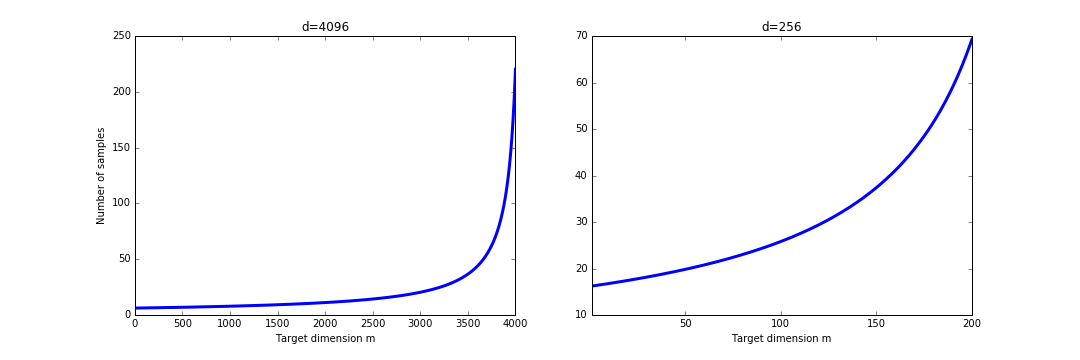}
\caption{The number of sample points at which the random projection method leads to computational savings.}\label{fig:savings}
\end{figure}

Having determined that random projections are worthwhile from a computational point of view, we next consider the distortion in the computation of persistent homology. Figure~\ref{fig:cat} shows the computation of persistent homology for $72$ images from the Columbia Object Image Library~\cite{nene1996columbia}. Each image has size $128\times 128$ and consists of the same object viewed from different perspectives (the display shows one of the images). The underlying topological structure is that of a circle, and hence the persistence barcode for the first homology shows one long bar. In the experiment, $200$ random projections were computed for target dimensions up to $400$, and the first persistent homology was computed on the projected data set. The proportion of projections for which we get an interleaving distance $c=1+\epsilon$ to the first persistent homology of the original data with $\epsilon<0.1$ was recorded. 
As proxy to the interleaving distance we computed the bottleneck distance~\cite{cohen2007stability} between the persistence diagrams at log-scale and set $c$ to be the exponential of this distance, as described in~\cite[Footnote 1]{sheehy2014persistent}. 
The computations were carried out using Vietoris-Rips complexes in the Python wrapper for Ripser~\cite{ctralie2018ripser,bauer2017ripser} and GUDHI~\cite{gudhi:urm}. 

\begin{figure}[H]
\centering
\includegraphics[width=1\textwidth]{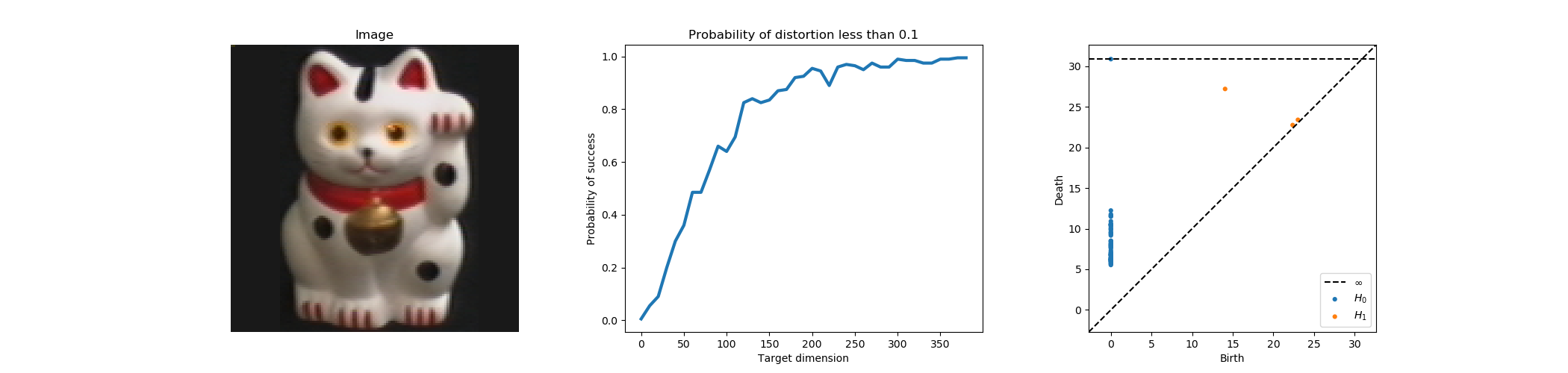}
\caption{Computing the persistent homology under random projections with distortion up to $\epsilon = 0.1$.}\label{fig:cat}
\end{figure}
\section{Conclusion}\label{sec:conclusion}
We have shown that randomized dimensionality reduction results related to the Johnson-Lindenstrauss Theorem carry over without change to the computation of persistent homology with respect to the Euclidean distance.
A consequence is that the target dimension can be chosen independently of the number of points, as a function of the Gaussian width. By relating the Gaussian width of the set of normalized differences of the points to the doubling dimension, the complexity reduction achievable from the randomized projection method is linked to that achievable by other methods, such as the construction of sparse filtrations. It is likely that the Gaussian width can also be used as a tool in the analysis of other reduction methods.
Another direction is suggested by the proof of the Kirszbraun intersection property in terms of Slepian's Lemma, given in Section~\ref{sec:kirszbraun}. It would be interesting to see if this approach generalizes to other
forms of projective clustering such as $k$-center clustering, by using Gordon's inequalities. Finally, a natural question is to what extent such randomized reductions are possible
for other notions of distance. 

\vskip6pt

\enlargethispage{20pt}


\dataccess{The code used to generate the examples in Section~\ref{sec:experiments} is available at \url{https://github.com/lotzma/compressivepersistent}}

\competing{There are not competing interests.}

\funding{The author received no specific funding for this work.}

\ack{I would like to thank Uli Bauer for point out the Kirszbraun intersection property and the reference~\cite{gromov1987monotonicity}, and Michael Kerber and Don Sheehy for feedback and some useful comments on a preliminary draft. 
The motivation to combine ideas from randomized dimension reduction with persistent homology arose from conversations with Jared Tanner during visits to the Alan Turing Institute, and I would like to acknowledge their support for making these visits possible. 
I am grateful to the anonymous referees for various helpful suggestions on improving the paper.}



\vskip2pc

\bibliographystyle{RS} 

\def\cprime{$'$}

\end{document}

%% file: macros.tex
\usepackage{amsmath,amsfonts,amsbsy,amsgen,amscd,mathrsfs,amssymb,amscd}
\usepackage{amsthm}
\usepackage{url}
\usepackage[UKenglish]{babel}
\usepackage{eurosym}
\usepackage{tikz}
\usetikzlibrary{matrix,arrows,shapes}
\tikzset{mynode/.style = {
    circle,
    minimum size=2pt,
    draw=black, fill=black}
}
\usepackage{microtype}
\usepackage{enumitem}
\usepackage{listings}
\definecolor{darkblue}{rgb}{0,0,.75}
\lstnewenvironment{PseudoCode}[1][]
{\lstset{language=Matlab,basicstyle=\small, keywordstyle=\color{darkblue},numbers=none,xleftmargin=.04\textwidth,mathescape,frame=single,#1}}
{}

\definecolor{dark-gray}{gray}{0.3}
\definecolor{dkgray}{rgb}{.4,.4,.4}
\definecolor{dkblue}{rgb}{0,0,.5}
\definecolor{medblue}{rgb}{0,0,.75}
\definecolor{rust}{rgb}{0.5,0.1,0.1}

\usepackage{graphicx}
\usepackage{booktabs,longtable,tabu} 
\usepackage{multirow} 
\usepackage{float}
\usepackage[T1]{fontenc}

\usepackage{times}
\usepackage{bm} 

\graphicspath{{figures/}}

\numberwithin{equation}{section} 

\providecommand{\mathbold}[1]{\bm{#1}}  

\renewcommand{\phi}{\varphi}

\newcommand{\e}{\varepsilon}

\renewcommand{\mid}{\mathrel{\mathop{:}}} 

\newcommand{\zerovct}{\vct{0}} 

\newcommand{\onemtx}{\bm{1}}

\providecommand{\mathbbm}{\mathbb} 

\newcommand{\R}{\mathbbm{R}}

\newcommand{\FF}{\mathbbm{F}}

\newcommand{\encl}{\operatorname{enc}}

\newcommand{\diff}[1]{\mathrm{d}{#1}}

\newcommand{\diam}{\mathrm{diam}}

\newcommand{\argmin}{\operatorname*{arg\; min}}

\newcommand{\Prob}{\mathbbm{P}}

\newcommand{\Expect}{\operatorname{\mathbb{E}}}

\newcommand{\normal}{N}

\newcommand{\vct}[1]{\mathbold{#1}}
\newcommand{\mtx}[1]{\mathbold{#1}}

\newcommand{\im}{\operatorname{im}}

\newcommand{\supp}[1]{\operatorname{supp}(#1)}

\newcommand{\Proj}{\ensuremath{\mtx{\Pi}}} 

\newcommand{\ip}[2]{\langle {#1}, {#2} \rangle}

\newcommand{\norm}[1]{\Vert {#1}\Vert}

\newcommand{\conv}{\operatorname{conv}}

\DeclareMathOperator{\rk}{rk}

